\documentclass[oneside,article]{memoir}
\usepackage{amsmath}         % align-like environments
\usepackage{amsthm,thmtools,thm-restate} % theorem-like environments
\usepackage{enumitem}        % enumerate-like environments

\usepackage{tikz}
  \usetikzlibrary{matrix,arrows,positioning,calc}
  
\usepackage{tikz-cd}

\usepackage{float} %to force figure into right place
\usepackage{caption, subcaption} % for subfigures

\usepackage[nobysame,alphabetic,initials]{amsrefs} % bibliography tools
\usepackage[hidelinks]{hyperref}                   % hyperlinks
\usepackage[capitalize,nosort]{cleveref}           % named references

\usepackage{amssymb}
\usepackage[mathscr]{euscript}
\usepackage{relsize}
\usepackage{stmaryrd}
\usepackage{stackengine}

\usepackage{indentfirst} % indent the first paragraph
\usepackage{multicol}
\usepackage{algorithm}
\usepackage{algorithmic}
\usepackage[inline,nomargin]{fixme} % fixme notes
  \fxsetup{targetlayout=color, status=draft,}

\usepackage{float} %to force figure into right place

\isopage[12]

\setlrmargins{*}{*}{1}
\checkandfixthelayout

\counterwithout{section}{chapter}
\usepackage{appendix}

  \newcommand{\iso}{\mathrel{\cong}}

  \declaretheorem[style=definition,within=section]{definition}
  \declaretheorem[style=definition,numberlike=definition]{example}
  \declaretheorem[style=definition,numberlike=definition]{remark}

  \declaretheorem[style=plain,numberlike=definition]{proposition}
  \declaretheorem[style=plain,numberlike=definition]{theorem}
  \declaretheorem[style=plain,numberlike=definition]{notation}

  \declaretheorem[style=plain,numbered=no,name=Theorem]{theorem*}

  \Crefname{corollary}{Corollary}{Corollaries}
  \Crefname{definition}{Definition}{Definitions}
  \Crefname{lemma}{Lemma}{Lemmas}
  \Crefname{proposition}{Proposition}{Propositions}
  \Crefname{remark}{Remark}{Remarks}
  \Crefname{theorem}{Theorem}{Theorems}
  \Crefname{notation}{Notation}{Notations}
  \Crefname{conjecture}{Conjecture}{Conjectures}
  \Crefname{example}{Example}{Examples}

  \newlist{axioms}{enumerate}{1}
  \Crefname{axiomsi}{}{}

  \newenvironment{tikzeq*}
  {
    \begingroup
    \begin{equation*}
    \begin{tikzpicture}[baseline=(current bounding box.center)]
  }
  {
    \end{tikzpicture}
    \end{equation*}
    \endgroup
    \ignorespacesafterend
  }

  \tikzset
  {
    diagram/.style=
    {
      matrix of math nodes,
      column sep={4.3em,between origins},
      row sep={4em,between origins},
      text height=1.5ex,
      text depth=.25ex
    },
    over/.style={preaction={draw=white,-,line width=6pt}},
    every to/.style={font=\footnotesize},
    inj/.style={right hook->},
    surj/.style={-{Latex[open]}},
    cof/.style={>->},
    fib/.style={->>},
  }

  \DeclareFontFamily{U}{mathx}{\hyphenchar\font45}

  \DeclareFontShape{U}{mathx}{m}{n}{
    <5> <6> <7> <8> <9> <10>
    <10.95> <12> <14.4> <17.28> <20.74> <24.88>
    mathx10}{}

  \DeclareSymbolFont{mathx}{U}{mathx}{m}{n}

  \DeclareFontFamily{U}{mathb}{\hyphenchar\font45}

  \DeclareFontShape{U}{mathb}{m}{n}{
    <5> <6> <7> <8> <9> <10>
    <10.95> <12> <14.4> <17.28> <20.74> <24.88>
    mathb10}{}

  \DeclareSymbolFont{mathb}{U}{mathb}{m}{n}

  \DeclareMathAccent{\widebar}{0}{mathx}{"73}

  \DeclareMathSymbol{\Rsh}{\mathrel}{mathb}{"E9}

  \DeclareFontFamily{U}{MnSymbolA}{}

  \DeclareFontShape{U}{MnSymbolA}{m}{n}{
    <-6> MnSymbolA5
    <6-7> MnSymbolA6
    <7-8> MnSymbolA7
    <8-9> MnSymbolA8
    <9-10> MnSymbolA9
    <10-12> MnSymbolA10
    <12-> MnSymbolA12}{}

  \DeclareSymbolFont{MnSyA}{U}{MnSymbolA}{m}{n}

  \DeclareMathSymbol{\twoheaddownarrow}{\mathrel}{MnSyA}{27}

  \newcommand{\MSC}[1]{%
    \let\thempfn\relax
    \footnotetext[0]{2020 Mathematics Subject Classification: #1.}
  }

\tikzstyle{vertex}=[circle, draw, minimum size=7pt, inner sep=0pt]

\renewcommand{\Im}{\operatorname{Im}}

\newcommand{\Graph}{\mathsf{Graph}} % category of graphs

\author{Jacob Ender \and Krzysztof Kapulkin} % alphabetical order (by last name)
\title{Towards fast computation of higher discrete homology}
\date{\today}

\begin{document}

  \maketitle

  \begin{abstract}
    We develop a new algorithm for computing the second discrete homology group of a graph which is much faster when compared to existing algorithms.
    To do so, we identify five basic shapes, which are quotient graphs of the 3-cube with the property that the injective maps from them detect all possible 2-boundaries in the singular chain complex computing discrete homology.
  \end{abstract}

%  \setlist[enumerate]{label=(\arabic*)}
  
% Add content here
\section*{Introduction}

Discrete homology groups of a graph, introduced in \cite{barcelo-capraro-white}, have proven over the previous decade to be both important and interesting graph invariants.
Their importance is underscored by the wealth of applications that they have found, ranging from pure mathematics, specifically areas such as hyperplane arrangements \cite{barcelo-laubenbacher} and matroid theory \cite{chalopin-chepoi-osajda}, to quite applied fields, including network analysis \cite{atkin:i,atkin:ii} and topological data analysis \cite{kapuklin-kershaw:data-analysis}.
There is currently very little that we understand about discrete homology groups, even of very small graphs.
Case in point: the Greene sphere, named after Curtis Greene, is a graph with 10 vertices and 16 edges whose discrete homology groups are known only up to degree $3$.

In recent years, significant effort was devoted to developing software to compute discrete homology groups.
The first major step in this direction was taken in \cite{barcelo-greene-jarrah-welker:comparison}, where machine-assisted computations were used to show that discrete homology sees qualitatively different information than path homology.
The algorithms developed there could compute the second discrete homology groups of most small graphs, but were generally unable to go beyond that.
The state of the art algorithm, independent of the homological degree, was developed several years later in a preprint \cite{kapuklin-kershaw:computations}, where several innovations were introduced, notably the pairing of lower-dimensional cubes to obtain higher-dimensional ones that subsequently unlocked numerous additional speedups.
This algorithm is now commonly referred to in literature as \emph{cubical}.

The aforementioned computational work has since informed many theoretical results and, conversely, was itself made possible by theoretical contributions.
For example, the work of Behrens and Jamil \cite{behrens-jamil} used the degree spectral sequence to provide a theoretical justification of the computation of the third discrete homology group of the Greene sphere from \cite{kapuklin-kershaw:computations}.
Likewise, \cite{carranza-kapulkin:discrete-homotopy-hypothesis} provided theoretical justification for several other computations of discrete homology and homotopy groups of graphs.
In the opposite direction, the theoretical work of \cite{carranza-kapulkin:cubical-setting,barcelo-greene-jarrah-welker:connections,kapulkin-mavinkurve,greene-welker-wille} provides the necessary mathematical foundation to ensure the correctness of algorithms developed in \cite{kapuklin-kershaw:computations}.
Put together, the interactions between the theoretical and the computational side of the field make discrete homotopy theory an active and vibrant area.

In \cite{ender-kapulkin:fast-h1}, a new algorithm, termed \emph{cellular}, for computing the first discrete homology group was proposed.
The algorithm is based on the observation that the first discrete homology group of a graph can be computed as the (classical) first discrete homology group of a cell complex (hence `cellular' algorithm) obtained from a graph by gluing 2-dimensional cells into every 3- and 4-cycle.
In other words, the question of computing discrete homology can be reduced to finding 3- and 4-cycles in a graph and doing a matrix computation, with a host of efficient methods available for both of these operations.
The present paper continues this line of research by outlining a strategy for finding higher-dimensional analogues of $3$- and $4$-cycles.

By definition, the basis of the singular chain complex $C_\bullet(G)$ associated to a graph $G$ is given, in degree $n$, by non-degenerate graph maps from the $n$-dimensional hypercube $Q^n \to G$.
Its grows doubly exponentially with $n$, making computations of the fourth discrete homology group not feasible in most cases.
The work of \cite{ender-kapulkin:fast-h1} can be understood in this context as wanting to replace the chain complex $C_\bullet(G)$ with a quasi-isomorphic one, call it $S_\bullet(G)$, with the property that, for every positive integer $n$, we have $\dim S_n(G) \ll \dim C_n(G)$.
Indeed, we can reinterpret the aforementioned result by saying that $S_1(G)$ is free on the edges and $S_2(G)$ is free on $3$- and $4$-cycles of $G$.

While we are unable to find $S_3(G)$ that extends this tentative chain complex, we are able to find a small subspace $S_3(G) \subseteq C_3(G)$ spanned by injective maps from five \emph{basic shapes} to $G$ with the property that the restriction of the differential $\partial_3 \colon C_3(G) \to C_2(G)$ to $S_3(G)$ has the same image (and, in particular, the same rank) as the differential itself.
Since the most expensive part of the computation of the second discrete homology group is identifying and reducing the matrix of this differential, we are able to improve computation time.
It was perhaps well-known to experts that these basic shapes should be drawn from quotients of the 3-dimensional hypercube; our contribution lies in identifying the specific five shapes that need to be considered.

\textbf{Organization.}
This paper is organized as follows.
In \cref{sec:prelims}, we review the necessary background on discrete homology.
In \cref{sec:fast_h2}, we describe our algorithm for computing the second discrete homology group.
Finally, in \cref{sec:experiments}, we report on the experiments that compare our algorithm with the now-standard cubical algorithm.

\textbf{Acknowledgments.}
We have greatly benefited from conversations with D.~Carranza and N.~Kershaw.
The idea behind \cref{def:n_shapes_cat} and the implementation of the algorithm in \cref{def:universal_list} were given to the authors by ChatGPT.

\section{Discrete homology} \label{sec:prelims}

In this section, we review the definition of discrete homology of graphs, following \cite{barcelo-capraro-white,barcelo-greene-jarrah-welker:comparison,barcelo-greene-jarrah-welker:vanishing}.
We begin by recalling the definition of a graph.

\begin{definition} \label{def:graph} \leavevmode
    \begin{enumerate}
        \item A \emph{graph} $G$ is a set $G_V$ together with a reflexive, symmetric relation $G_E \subseteq G_V \times G_V$.
        The set $G_V$ is called the set of \emph{vertices} of $G$, and the relation $G_E$ is called the set of \emph{edges} of $G$. 
        \item For graphs $G$ and $H$, a function $f \colon G_V \to H_V$ is a \emph{graph map} if $f$ preserves the adjacency relation. 
        \item Graphs and graph maps form a category denoted $\Graph$.
    \end{enumerate}
\end{definition}

\begin{notation}
If $v, w \in G_V$, the notation $v \sim w$ means that $(v,w) \in G_E$.
\end{notation}

\begin{example} \label{def:graph_exs} \leavevmode
    \begin{enumerate}
        \item The graph $I_n$ has vertex set $\{ 0, 1, \dots, n \}$, with edges $(i, i+1)$ for $0 \leq i \leq n-1$. 
        \item The graph $C_n$ is the graph $I_n$, quotiented by the relation identifying vertex $0$ with vertex $n$.
        \item For an integer $n \geq 1$, the \emph{complete graph} $K_n$ is a graph with $n$ vertices and an edge between each pair of vertices $v$ and $w$.
    \end{enumerate}

    Examples of each of these graphs are depicted below.

    \begin{center}
            \begin{tikzpicture}[scale=1,
              every node/.style={circle, fill=black, inner sep=1.5pt},
              every path/.style={thick}
            ]
        
                % Common y baseline at y=0
                % 1. I_n (n=4)
                \begin{scope}[shift={(-1,0)}]
                  \foreach \i in {0,...,4} {
                    \node (a\i) at (\i,0) {};
                  }
                  \foreach \i in {0,...,3} {
                    \draw (a\i) -- (a\the\numexpr\i+1\relax);
                  }
                  \node[draw=none, fill=none, below=0.5cm of a2, font=\small] {The graph $I_4$.};
                \end{scope}

                % 3. C_n (n=5) - upright pentagon
                \begin{scope}[shift={(6.5,0)}]
                  % vertices centered so middle matches y=0 baseline
                  \foreach \i in {0,...,4} {
                    \node (c\i) at ({90 - 72*\i}:1.5) {};
                  }
                  \foreach \i in {0,...,4} {
                    \pgfmathtruncatemacro{\j}{mod(\i+1,5)}
                    \draw (c\i) -- (c\j);
                  }
                  % Place label at a fixed offset from center
                  \node[draw=none, fill=none, below=2cm of c0, font=\small] {The graph $C_5$. };
                \end{scope}
    
                % K_4
                \begin{scope}[shift={(11.5,0)}]
                  \foreach \i in {0,...,3} {
                    \node (k\i) at ({45 - 90*\i}:1.5) {};
                  }
                  \foreach \i in {0,...,3} {
                    \foreach \j in {0,...,3} {
                      \ifnum\i<\j
                        \draw (k\i) -- (k\j);
                      \fi
                    }
                  }
                  \node[draw=none, fill=none, font=\small] at (0, -1.7) {The graph $K_4$.};
                \end{scope}
        
            \end{tikzpicture}
        \end{center}
\end{example}

There is more than one notion of a product of two graphs, but for the purposes of defining discrete homology, we need only consider the box product. 

\begin{definition} \label{def:boxprod}
Given graphs $G$ and $H$, their \emph{box product} $G \square H$ has vertex set $G_V \times H_V$, and edges are pairs $((u, v), (u', v'))$, where either $u = u'$ and $v \sim v'$ or $u \sim u'$ and $v = v'$. 
\end{definition}

\begin{example} \leavevmode
    \begin{center}
    
        \begin{tikzpicture}[scale=1,
          every node/.style={circle, fill=black, inner sep=1.5pt},
          every path/.style={thick}
        ]
    
          \begin{scope}[shift={(-1.5,0)}]
          \foreach \x in {0,...,2} {
            \foreach \y in {0,...,2} {
              \node (b\x\y) at (\x, \y) {};
            }
          }
          % horizontal edges
          \foreach \x in {0,...,1} {
            \foreach \y in {0,...,2} {
              \draw (b\x\y) -- (b\the\numexpr\x+1\relax\y);
            }
          }
          % vertical edges
          \foreach \x in {0,...,2} {
            \foreach \y in {0,...,1} {
              \draw (b\x\y) -- (b\x\the\numexpr\y+1\relax);
            }
          }
          \node[draw=none, fill=none, below=0.5cm of b11, font=\small] {The graph $I_2^{\square 2}$.};
          \end{scope}
    
           \begin{scope}[shift={(3.5,0)}]
                  
                \node(A) at (0.5,0) {};
                \node(B) at (2,0) {};
                \node(C) at (2,1.5) {};
                \node(D) at (0.5,1.5) {};
    
                \node(E) at (1,0.5) {};
                \node(F) at (2.5,0.5) {};
                \node(G) at (2.5,2) {};
                \node(H) at (1,2) {};
            
                \draw (A) -- (B) -- (C) -- (D) -- (A);
                \draw (E) -- (F) -- (G) -- (H) -- (E);
    
                \draw (A) -- (E);
                \draw (B) -- (F);
                \draw (C) -- (G);
                \draw (D) -- (H);
                
                  \node[draw=none, fill=none, font=\small] at (1.5, -0.78) {The graph $I_1^{\square 3}$.};
            \end{scope}
            
        \end{tikzpicture}
    
    \end{center}
\end{example}

Now that we have a notion of graph products, we can build up the structures that give us a chain complex to compute the homology of. The structures that we use are \emph{cubes} in a graph, which are given by the following two definitions.

\begin{definition} \label{def:discrete_n_cube}
    The \emph{discrete n-cube} is the graph $I_1^{\square n}$. Vertices are labeled $(x_1, \dots, x_n)$, where $x_i \in \{0,1\}$ for $0 \leq i \leq n$. For brevity, we denote the discrete $n$-cube by $Q^n$. 
\end{definition}

\begin{definition} \label{def:sing_n_cube}
    If $G$ is a graph, then a \emph{singular n-cube} in $G$ is a graph map $A \colon Q^n \to G$. 
\end{definition}

Given a singular $n$-cube, one can ask about its faces. In particular, defining faces allows us to track whether or not a singular $n$-cube is degenerate.

\begin{definition} \label{def:face_maps}
    Fix a graph $G$ and a singular $n$-cube $A \colon Q^n \to G$, where $n \geq 1$. For $1 \leq i \leq n$, define the singular $(n-1)$-cubes $\delta_i^{+}A$ and $\delta_i^{-}A$ as follows.
    \begin{align*}
        \delta_i^+A(x_1,\dots,x_{n-1}) &:= A(x_1, \dots, x_{i-1}, 1, x_i, \dots, x_{n-1}) \\
        \delta_i^-A(x_1,\dots,x_{n-1}) &:= A(x_1, \dots, x_{i-1}, 0, x_i, \dots, x_{n-1}). 
    \end{align*}
    The cubes $\delta_i^+A$ and $\delta_i^-A$ are called the $i^{\text{\emph{th}}}$ \emph{positive and negative face maps}, respectively. If $\delta_i^+A = \delta_i^-A$ for some $i$, then we call $A$ \emph{degenerate}. Otherwise, we call $A$ \emph{non-degenerate}. 
\end{definition}

We now have everything we need to give a definition of discrete homology. We begin with the following definition.

\begin{definition} \label{def:modules}
    For a graph $G$ and $n \geq 0$, let $L_n(G)$ be the free $\mathbb{Z}/2$-vector space generated by all singular $n$-cubes in $G$. Let $D_n(G)$ be the free $\mathbb{Z}/2$-vector space generated by all degenerate $n$-cubes in $G$. Finally, let $C_n(G) = L_n(G) / D_n(G)$, so that $C_n(G)$ is the free submodule of $L_n(G)$ generated by all non-degenerate $n$-cubes in $G$. 
\end{definition}

We are now ready to give the definition of the boundary of a singular $n$-cube, and finally of the discrete homology groups of a graph.

\begin{definition} \label{def:boundary}
    Fix a graph $G$. For $n \geq 1$ and a singular $n$-cube $A$ in $G$, define the \emph{boundary} of $A$ as the formal sum
    \[
    \partial_n(A) := \sum_{i=1}^n (-1)^i (\delta_i^-A - \delta_i^+A). 
    \]
    We can extend this map linearly to a map $\partial_n \colon L_n(G) \to L_{n-1}(G)$. 
\end{definition}

We set $\partial_0$ to be the trivial map. In \cite{barcelo-capraro-white}, it was shown that for any $n \geq 0$, we have $\partial_n \circ \partial_{n+1} = 0$ and $\partial_n[D_n(G)] \subseteq D_{n-1}(G)$. Thus if we set $C_{-1}(G) = (0)$, we obtain a chain complex $C(G) = (C_{\bullet}, \partial_{\bullet})$.  

\begin{definition} \label{def:homology_def}
    For $n \geq 0$, the \emph{n-th discrete homology group of G} is $\mathcal{H}_n(G) = \operatorname{Ker}(\partial_n) / \Im(\partial_{n+1})$. 
\end{definition}

\section{Towards computing higher discrete homology groups} \label{sec:fast_h2}

Computing discrete homology directly from the definition is virtually impossible by hand.
Case in point: the second and third discrete homology group of the Greene sphere, a graph with only 10 vertices and 16 edges, was only identified using machine computations on large compute clusters, while the fourth and higher homology groups of that graph remain unknown.
It is however possible to speed up computations of the first discrete homology group of a graph $G$, which was already rather inefficient.
To do so, we might instead compute cellular homology of a CW-complex $X_G$ associated to $G$.
The $1$-skeleton of $X_G$ is given by $G$ and we glue in a single $2$-cell for every (simple) $3$- and $4$-cycle in $G$.

\begin{theorem}[{\cite[Thm.~3.4]{ender-kapulkin:fast-h1}, cf.~\cite[Prop.~5.12]{barcelo-kramer-laubenbacher-weaver}}]\label{thm:hlgy_correspondence}
    For a connected graph $G$, we have $\mathcal{H}_1(G) \iso H_1(X_G)$, where $H_1$ denotes the singular homology of a topological space. \qed
\end{theorem}

This construction was exploited in \cite{ender-kapulkin:fast-h1} to obtain an algorithm for computing the first discrete homology group that outperforms existing algorithms \cite{kapuklin-kershaw:computations}.
The key insight here is that the vector space $C_2(G)$ can be replaced by a vector space $U_2(G)$ of a much smaller dimension.

When trying to generalize this to higher dimensions, one must answer the following question: what are the higher-dimensional analogues of $3$- and $4$-cycles?
Given an answer to this question, one can try to build a CW-complex by gluing $n$-cells into each such \emph{basic shape}.
The $3$- and $4$-cycles in a graph $G$ need to be understood here in relation to the object they are replacing, namely maps $Q^2 \to G$.
More precisely, both $3$- and $4$-cycles arise as quotients of $Q^2$, the notion that we now recall.

\begin{definition}
    A graph map $f \colon G \to H$ is a \emph{quotient map} if it is surjective on edges.
\end{definition}

We note that a graph map surjective on edges is automatically surjective on vertices, since we assume that our graphs are reflexive.

An important way of constructing quotient graphs is by taking factorizations of graph maps.
More precisely, given a graph map $f \colon G \to H$, we can factor it through its image as follows:
    \[
        \begin{tikzcd}
    	G && H \\
    	& {\Im f}
    	\arrow["f", from=1-1, to=1-3]
    	\arrow["{q_f}"', two heads, from=1-1, to=2-2]
    	\arrow["{i_f}"', hook, from=2-2, to=1-3]
        \end{tikzcd}
    \]
Explicitly, we define an equivalence relation $\sim_f$ on the vetex set of $G$ by declaring $v \sim_f v'$ if and only if $f(v) = f(v')$.
The set of vertices of $\Im f$ is then the quotient $G_V/\sim_f$ with the edge relation given by $[v] \sim [w]$ if there are $v' \in [v]$ and $w' \in [w]$ such that $v' \sim_f w'$, or equivalently, $f(v') = f(w')$.
In the factorization above, the graph map $q_f \colon G \to \Im f$, given by $q_f(v) = [v]$, is a quotient map and $i_f \colon \Im f \to H$, given by $i_f[v] = f(v)$ is a monomorphism (i.e., injective on both vertices and edges).
It is easy to see that the factorization of $f \colon G \to H$ into a quotient map followed by a monomorphism is unique up to a unique isomorphism of factorizations.

Returning to the question at hand, since $3$- and $4$-cycles are quotients of $Q^2$, one might reasonably expect that the set of basic shapes in dimension $3$ will be a set of quotients of $Q^3$.
However, things become slightly more complicated, since one needs to keep track not only of the quotient graph itself but also of the relation by which the quotient was formed.
To this end, we define the category of shapes.

\begin{definition} \label{def:n_shapes_cat}
    For a nonnegative integer $n$, define the \emph{category of $n$-shapes}, denoted $\mathcal{S}_R(n)$, with
    \begin{itemize}
        \item objects given by pairs $(P,q)$ such that $q \colon Q^n \to P$ is a quotient map;
        \item a morphism from $(P,q)$ to $(P',q')$  is a graph map $f \colon P \to P'$ making the following diagram commute:
     \[
        \begin{tikzcd}
        	& {C_n(Q^n; R)} & \\
        	{C_n(P; R)} && {C_n(P'; R)} \\
        	{C_{n-1}(P; R)} && {C_{n-1}(P'; R)}
        	\arrow["{q_*}"', from=1-2, to=2-1]
        	\arrow["{q'_*}", from=1-2, to=2-3]
        	\arrow["{\partial_n}"', from=2-1, to=3-1]
        	\arrow["{\partial'_n}", from=2-3, to=3-3]
        	\arrow["{f_*}", from=3-1, to=3-3]
        \end{tikzcd}
    \]
    \end{itemize} 
\end{definition}

\begin{remark}
  The definition of morphism of $n$-shapes might seem esoteric at first.
  Note that any graph map $f \colon P \to P'$ making the following triangle commute
  \[
  \begin{tikzcd}
      & Q^n \ar[ld, "q", swap] \ar[rd, "q'"] & \\
      P \ar[rr, "f"] & & P'
  \end{tikzcd}
  \]
  is a morphism of $n$-shapes from $(P,q)$ to $(P', q')$.
  That is, every morphism between quotient maps in the slice category $Q^n/\Graph$ is a morphism in $\mathcal{S}_R(n)$; the converse is not necessarily true, although we do not know of an example.
  Consequently, for $n$-shapes $(P, q)$ and $(P', q')$, we have inclusions:
 \[
  \left\{ 
  \begin{array}{cc}
  \text{isomorphisms} \\
   (P, q) \iso (P', q' )\\
   \text{in } Q^n/\Graph
   \end{array}
  \right\}
  \subseteq 
  \left\{
    \begin{array}{cc}
  \text{isomorphisms} \\
   (P, q) \iso (P', q' ) \\
   \text{in } \mathcal{S}_R(n)
   \end{array}
  \right\}
  \subseteq
    \left\{
  \begin{array}{cc}
  \text{isomorphisms} \\
  P \iso P' \\
   \text{in } \Graph
   \end{array}
  \right\}
  \]
\end{remark}

\begin{notation}
  We write $S_R(n)$ for the set of isomorphism classes of objects in $\mathcal{S}_R(n)$.  
\end{notation}

Although the cardinality of $S_R(n)$ is always finite, it might be computationally difficult to find it.
Such a search would naturally involve the $2^n$-th \emph{Bell number} $B_{2^n}$, which counts partitions of the set with $2^n$ elements.
These correspond to the number of isomorphism classes of objects in the full subcategory of $Q^n/\Graph$ spanned by quotient maps.
The Bell numbers nevertheless provide some estimate for the number of isoclasses $S_R(n)$, namely: $B_1 = 1$, $B_2 = 2$, $B_8 = 4140$, and $B_{16} = 10480142147$, thus yielding already uncomputable results for $n = 4$.

To see how shapes can aid computations of discrete homology, we need the notion of universality.
    
% Since we demand that the quotient map $q$ from any object $(P,q) \in S_R(n)$ be surjective, there are finitely many isomorphism classes of objects in $\mathcal{S}(n)$.
% In particular, to determine all such isomorphism classes, we must inspect a number of quotient maps at most equal to the number of partitions of a set with $2^n$ elements, which is the Bell number $B_{2^n}$. This number explodes with $n$, so it is impractical to check every isomorphism class of $n$-shapes, even for $n = 4$. Fortunately, it is still possible to find a subset of isomorphism classes of $S_R(n)$ that, in some sense, generates $\Im(\partial_n)$. To achieve this goal, we must first develop the notion of universality.

\begin{definition} \label{def:universal_list} \leavevmode
    \begin{itemize}
        \item A subset $T \subseteq S_R(n)$ is \emph{universal in a graph $G$} if
        \[
        \Im\big(\partial_n \colon C_n(G; R) \to C_{n-1}(G; R)\big) = \operatorname{span}\left\{ \partial_n(iq) \mid (P,q) \in T \text{ and } i \colon P \hookrightarrow G \text{ injective}\right\}.
        \]
        \item A subset $T \subseteq S_R(n)$ is \emph{universal} if it is universal in $G$ for all graphs $G$. 
    \end{itemize}
\end{definition}

\begin{example}
    The proof of \cite[Thm.~3.4]{ender-kapulkin:fast-h1} shows that the set $\{C_3, C_4\}$ is a universal subset of $2$-shapes with any choice of quotient map $Q^2 \to C_3$.
\end{example}

% The following two results justify the use of a universal subset of $n$-shapes.

% \begin{proposition} \label{prop:reduced_image}
%     If $T \subseteq S_R(n)$ is universal, then $\Im(\partial_n) = \Im\left(\partial_n|_{\operatorname{span}T} \right)$.
% \end{proposition}

% \begin{proof}
%     This follows immediately from the definition of universality and linearity of $\partial_n$.
% \end{proof}

% \begin{corollary}
%     $\operatorname{rank}(\partial_n) = \operatorname{rank}\left(\partial_n|_{\operatorname{span}T}\right)$. \qed
% \end{corollary}

If we can find a small, universal set $T$ of $n$-shapes for some nonnegative integer $n$, then we need only look for injective maps from each shape in $T$ to a graph $G$.
However, \cref{def:universal_list} requires us to iterate over \emph{all} graphs $G$, which is impossible.
The following proposition removes this obstacle, allowing us to decide if a set $T$ is universal only by comparing elements of $T$ to elements of the finite set $S_R(n)$. 

\begin{proposition} \label{prop:local_universality}
    A subset $T \subseteq S_R(n)$ is universal if and only if it is universal for every $P$ such that $(P,q) \in S_R(n)$. 
\end{proposition}

\begin{proof}
    The forward implication is clear. For the reverse inclusion, fix an arbitrary graph $G$ and let $x \in \Im(\partial_n^G \colon C_n(G; R) \to C_{n-1}(G; R))$, the image of the $n$-th boundary map for $G$. Then, $x = \partial_n^G(u)$ for some singular $n$-cube $u \colon I^n \to G$. 
    Factor $u$ as
    \[
        \begin{tikzcd}
    	{I^n} && G \\
    	& {P}
    	\arrow["u", from=1-1, to=1-3]
    	\arrow["{q}"', two heads, from=1-1, to=2-2]
    	\arrow["{i}"', hook, from=2-2, to=1-3]
        \end{tikzcd}
    \]
    Then, $x = \partial_n^G(iq) = i_*(\partial_n^Pq)$, as chain maps commute with boundaries. 
    Universality of $T$ with respect to each graph $P \in S_R(n)$ then gives
    \begin{align*}
        x &= i_*\left( \sum_{\substack{T = (P', q') \in T \\ i' \colon P' \hookrightarrow P\text{ injective}}} \alpha_{i', T} \partial_n^P(i' \circ q')\right) \\ 
        &=  \sum_{\substack{T = (P', q') \in T \\ i' \colon P' \hookrightarrow P\text{ injective}}} \alpha_{i', T} \partial_n^G(i \circ i' \circ q')
    \end{align*}
    again using linearity and the fact that chain maps commute with boundaries. Observe, however, that $i \circ i' \circ q' \in \{iq \mid (P,q) \in  T \text{ and } i \colon P \hookrightarrow G \text{ injective}\}$, so we are done. 
\end{proof}

\cref{prop:local_universality} gives us a finite computer-verifiable procedure for deciding whether a set $T$ of $n$-shapes is universal.
Namely, we can compute all possible isomorphism classes of $n$-shapes and perform finitely many membership tests on the $n$-shapes in $S_R(n)$.  

It remains to see how to construct a small, universal subset of $n$-shapes given the full list $S_R(n)$. One way to accomplish this is by means of a greedy algorithm.

\begin{definition} \label{alg:greedy_reduction}
    Let $N = \# S_R(n)$.
    Order the elements $t_1, \dots, t_N$ of $S_R(n)$, where $t_i = (P_i, q_i)$.
    Let $F_0 = \varnothing$ and, for $k \geq 1$, let
    \[
    F_{k} = \begin{cases}
        F_{k-1}, & \Im(\partial_n|_{\operatorname{span}(F_{k-1} \cup \{t_k\})}) = \operatorname{span} \left\{ \partial_n(iq) \mid (P,q) \in F_{k-1} \text{ and } i \colon P \hookrightarrow P_k \text{ injective}\right \} \\ 
        F_{k-1} \cup \{t_k\} & \text{otherwise}.
    \end{cases}
    \]
    The set $F_N$ is called the \emph{reduced list} for the chosen ordering. 
\end{definition}

\begin{proposition}
    For every ordering of $S_R(n)$, the corresponding reduced list is universal.
\end{proposition}

\begin{proof}
    Let $F_N$ be the reduced list, and fix any $n$-shape $t_k = (P_k, q_k) \in S_R(n)$. If $t_k \in F_N$, then trivially 
    \[
    \partial_n(q_k) \in \operatorname{span} \{ \partial_n(iq) \mid (P,q) \in F_{N}, i \colon P \hookrightarrow P_k \}
    \]by setting $P = P_k$ and $i = \text{id}_{P_{k}}$. If $t_k \not\in F_N$, then the greedy algorithm skipped it at some step $k_0$. For this to be the case, we must have already had 
    \[
    \partial_n(q_k) \in \operatorname{span} \{ \partial_n(iq) \mid (P,q) \in F_{k_0-1}, i \colon P \hookrightarrow P_k \}.
    \]By definition, $F_{k-1} \subseteq F_k$ for all $k \geq 1$. Thus, $F_{k_0-1} \subseteq F_N$, so we must have 
    \[
    \partial_n(q_k) \in \operatorname{span} \{ \partial_n(iq) \mid (P,q) \in F_N, i \colon P \hookrightarrow P_k \}.
    \]
    By \cref{prop:local_universality}, $F_N$ is universal.
\end{proof}

Note that while this algorithm provides a procedure for obtaining a reduced universal list of $n$-shapes, it does not guarantee that the reduced list is minimal. Different orderings of $S_R(n)$ may produce different reduced lists $F_N$. 

Up to this point, we have developed a procedure that produces a reduced, universal list of $n$-shapes, for each $n \geq 0$. We give the $n$-shapes in such a list a name. 

\begin{definition} \leavevmode
    \begin{itemize}
        \item For each $n \geq 0$ and an arbitrary ordering of the elements of $S_R(n)$, denote by $T_R(n)$ the reduced universal list of $n$-shapes obtained by applying the algorithm given in \cref{alg:greedy_reduction}.
        \item Fix a graph $G$, an integer $n \geq 0$, and an arbitrary ordering of $S_R(n)$.
        Denote by $U_n(G; R)$ the free $R$-module on the set of all injective graph maps $\varphi \colon P \hookrightarrow G$, where $(P,q) \in T_n$.
    \end{itemize}  
\end{definition}

Using these new definitions, our work so far can be distilled into the following result.

\begin{proposition} \label{prop:basic_comm_diag}
    For a graph $G$, an integer $n \geq 0$ and an arbitrary ordering of the elements of $S_R(n)$, we have the commutative diagram
   
    \[
        \begin{tikzcd}
        	{C_{n+1}(G; R)} & {C_n(G; R)} & {C_{n-1}(G; R)} \\
        	{U_{n+1}(G; R)}
        	\arrow["{\partial_{n+1}^C}", from=1-1, to=1-2]
        	\arrow["{\partial_n^C}", from=1-2, to=1-3]
        	\arrow[hook', from=2-1, to=1-1]
        	\arrow["{\partial_{n+1}^U}"', from=2-1, to=1-2]
        \end{tikzcd}
    \]
   where the operator $\partial_\bullet^C$ is the usual boundary operator, and $\partial_\bullet^U$ is the usual boundary operator restricted to $\operatorname{span}(U_\bullet(G; R))$. Moreover, we have that $\Im\partial_{n+1}^U = \Im\partial_{n+1}^C$, for $n \geq 0$. \qed 
\end{proposition}

\begin{remark} \label{ex:diff_boundaries}
  In contrast to the case $n=2$, in general, we cannot simply consider subgraphs of $G$ that are abstractly isomorphic to some quotient of $Q^3$.
  To see this, consider the following singular 3-cubes $A, A' \colon Q^3 \to C_4$, where in each diagram, the label on each vertex of $Q^3$ is the image of that vertex in $C_4$, so each vertex is labelled with an integer between $0$ and $3$:
        \[
        \begin{tikzpicture}[scale=1, every node/.style={circle, fill=black, inner sep=1.5pt}]
            
            \node[label=below left:0] (A) at (0,1) {};
            \node[label=below right:3] (B) at (1.5,1) {};
            \node[label=above:2] (C) at (1.5,2.5) {};
            \node[label=above left:1] (D) at (0,2.5) {};

            \node (E)[label=above left:1] at (0.5,1.5) {};
            \node (F)[label=below right:2] at (2,1.5) {};
            \node (G)[label=above right:3] at (2,3) {};
            \node (H)[label=above left:0] at (0.5,3) {};
        
            \draw (A) -- (B) -- (C) -- (D) -- (A);
            \draw (E) -- (F) -- (G) -- (H) -- (E);

            \draw (A) -- (E);
            \draw (B) -- (F);
            \draw (C) -- (G);
            \draw (D) -- (H);

            \node[label=below left:0] (A) at (4,1) {};
            \node[label=below right:1] (B) at (5.5,1) {};
            \node[label=above:2] (C) at (5.5,2.5) {};
            \node[label=above left:3] (D) at (4,2.5) {};

            \node (E)[label=above left:3] at (4.5,1.5) {};
            \node (F)[label=below right:0] at (6,1.5) {};
            \node (G)[label=above right:1] at (6,3) {};
            \node (H)[label=above left:2] at (4.5,3) {};
        
            \draw (A) -- (B) -- (C) -- (D) -- (A);
            \draw (E) -- (F) -- (G) -- (H) -- (E);

            \draw (A) -- (E);
            \draw (B) -- (F);
            \draw (C) -- (G);
            \draw (D) -- (H);   
        \end{tikzpicture}\]
        One easily checks that $A$ and $A'$ are graph maps, and that their quotient graphs $\Im A$ and $\Im A'$ are both isomorphic to $C_4$.
        Moreover, the maps $i_A$ and $i_{A'}$ have equal images in $C_4$.
        However, a simple computation shows that $\partial_3(A) = \delta_1^+ A - \delta_1^- A$, but $\partial_3(A')$ has six non-vanishing terms, so neither one is a linear combination of the other, while they both need to appear in the image of $\partial^U_3$.
        To keep track of this information we are therefore forced to include the quotient map $q \colon Q^3 \to P$ as part of the data of a shape.
\end{remark}

Of course, our results alone do not fully solve the basic shapes conjecture outlined in the introduction.
It is at least unclear whether the $R$-modules $\{U_n(G; R)\}_{n \geq 0}$ form a chain complex $(U_\bullet, \partial_\bullet)$ along with a quasi-isomorphism $U_n(G; R) \hookrightarrow C_n(G; R)$.
Despite these obstacles, the local criterion established by \cref{prop:basic_comm_diag} allows us to compute $\mathcal{H}_n(G; R)$ provided we can compute $U_{n+1}(G; R), C_n(G; R)$ and $C_{n-1}(G; R)$.
As we saw at the beginning of the section, the cardinality of these objects blow up exponentially with $n$, making computation intractable for $n \geq 3$.
However, in the case $n = 2$, the sets are still small enough that we can find a set of basic 3-shapes and use it to compute $\mathcal{H}_2(G)$.

\section{Experiments} \label{sec:experiments}

As an experiment, we compute a reduced set $T(3)$ of 3-shapes using the greedy algorithm outlined in \cref{alg:greedy_reduction}. 
In this section, all computations will be done over $R = \mathbb{Z}/2$ and we will omit the subscript for brevity.
To do so, we first enumerate each element of $S(3)$ by brute force, and then use the greedy algorithm to reduce this list to a small, universal set of 3-shapes. For this experiment, we order the elements of $S(3)$ in ascending order by the number of vertices in the associated quotient of $Q^3$, with ties being broken arbitrarily. As remarked before, choosing a different ordering may well produce a different $T(3)$.
In this case, we recover the reduced, universal set of five 3-shapes depicted below (where we omit the associated boundaries for brevity).
\[
\begin{tikzpicture}[scale=1, every node/.style={circle, fill=black, inner sep=1.5pt}]
  % E_1 (n=4, m=5) - vertical diamond with horizontal middle edge
  \node (T1_1) at (-1.000,1.100) {};   % top
  \node (T1_2) at (-1.700,0.000) {};   % left
  \node (T1_3) at (-0.300,0.000) {};   % right
  \node (T1_4) at (-1.000,-1.100) {};  % bottom
  \draw (T1_1) -- (T1_2);
  \draw (T1_1) -- (T1_3);
  \draw (T1_2) -- (T1_3);
  \draw (T1_2) -- (T1_4);
  \draw (T1_3) -- (T1_4);
  \node[fill=none, below, font=\footnotesize] at (-1.000,-1.400) {$E_1$};
  % E_2 (n=5, m=6) - K_{2,3} with crossings
  \node (T4_1) at (2.381,-0.397) {};
  \node (T4_2) at (1.300,0.000) {};
  \node (T4_3) at (3.219,0.705) {};
  \node (T4_4) at (3.219,-0.705) {};
  \node (T4_5) at (2.381,0.397) {};
  \draw (T4_1) -- (T4_2);
  \draw (T4_1) -- (T4_3);
  \draw (T4_1) -- (T4_4);
  \draw (T4_2) -- (T4_5);
  \draw (T4_3) -- (T4_5);
  \draw (T4_4) -- (T4_5);
  \node[fill=none, below, font=\footnotesize] at (2.500,-1.400) {$E_2$};
  % E_3 (n=6, m=9) - triangular prism 
  \node (T17_1) at (5.700,0.500) {};   % front top
  \node (T17_3) at (4.900,-0.700) {};  % front bottom-left
  \node (T17_4) at (6.500,-0.700) {};  % front bottom-right
  \node (T17_2) at (6.300,1.000) {};   % back top
  \node (T17_5) at (5.500,-0.200) {};  % back bottom-left
  \node (T17_6) at (7.100,-0.200) {};  % back bottom-right
  \draw (T17_1) -- (T17_2);
  \draw (T17_1) -- (T17_3);
  \draw (T17_1) -- (T17_4);
  \draw (T17_2) -- (T17_5);
  \draw (T17_2) -- (T17_6);
  \draw (T17_3) -- (T17_4);
  \draw (T17_3) -- (T17_5);
  \draw (T17_4) -- (T17_6);
  \draw (T17_5) -- (T17_6);
  \node[fill=none, below, font=\footnotesize] at (6.000,-1.400) {$E_3$};
  % E_4 (n=7, m=11) - cube with diagonal
  \node (T72_1) at (8.450,-0.850) {};
  
  \node (T72_3) at (8.450,0.850) {};
  \node (T72_4) at (10.150,0.850) {};
  \node (T72_5) at (8.933,-0.552) {};
  \node (T72_6) at (10.632,-0.552) {};
  \node (T72_7) at (8.833,1.147) {};
  \node (T72_8) at (10.532,1.147) {};

  \draw (T72_1) -- (T72_3);
  \draw (T72_1) -- (T72_5);

  \draw (T72_3) -- (T72_4);
  \draw (T72_3) -- (T72_7);
  \draw (T72_4) -- (T72_8);
  \draw (T72_5) -- (T72_6);
  \draw (T72_5) -- (T72_7);
  \draw (T72_6) -- (T72_8);
  \draw (T72_7) -- (T72_8);

  \draw (T72_1) -- (T72_4);
  \draw (T72_6) -- (T72_4);
 
  \node[fill=none, below, font=\footnotesize] at (9.500,-1.400) {$E_4$};
  % E_5 (n=8, m=12) - 3-cube
  \node (T82_1) at (12.150,-0.850) {};
  \node (T82_2) at (13.850,-0.850) {};
  \node (T82_3) at (12.150,0.850) {};
  \node (T82_4) at (13.850,0.850) {};
  \node (T82_5) at (12.533,-0.552) {};
  \node (T82_6) at (14.232,-0.552) {};
  \node (T82_7) at (12.533,1.147) {};
  \node (T82_8) at (14.232,1.147) {};
  \draw (T82_1) -- (T82_2);
  \draw (T82_1) -- (T82_3);
  \draw (T82_1) -- (T82_5);
  \draw (T82_2) -- (T82_4);
  \draw (T82_2) -- (T82_6);
  \draw (T82_3) -- (T82_4);
  \draw (T82_3) -- (T82_7);
  \draw (T82_4) -- (T82_8);
  \draw (T82_5) -- (T82_6);
  \draw (T82_5) -- (T82_7);
  \draw (T82_6) -- (T82_8);
  \draw (T82_7) -- (T82_8);
  \node[fill=none, below, font=\footnotesize] at (13.000,-1.400) {$E_5$};
\end{tikzpicture}
\]
The cell $E_5$ is simply $Q^3$.
The cell $E_4$ is $Q^3$ with one edge contracted. 
The triangular prism $E_3$ is obtained by identifying two pairs of adjacent vertices that form parallel edges of $E_3$.
The cell $E_2$ is isomorphic (as graphs) to the complete bipartite graph $K_{2,3}$.
Finally, the cell $E_1$ is isomorphic to the complete tripartite graph $K_{1,1,2}$. 

We run the following experiment using this set of basic 3-shapes.

\begin{enumerate}
    \item Pre-compute the set $S(3)$ by enumerating every set partition of 8 vertices with nondegenerate quotient maps and nonzero boundaries.
    \item Use the greedy algorithm from \cref{alg:greedy_reduction} with the aforementioned ordering of the elements of $S(3)$ to pre-compute $T(3)$.
    \item For an input graph $G$, perform the following steps.
    \begin{enumerate}
        \item Let $C_1(G)$ consist of all non-loop edges in $G$;
        \item Let $C_2(G)$ consist of all nondegenerate singular 2-cubes in $G$;
        \item For each element $t = (P, q) \in T(3)$, enumerate all injective graph maps $\varphi \colon P \hookrightarrow G$. This computes $U_3(G)$. 
        \item Compute the matrices of the boundary operators $\partial^C_2$ and $\partial^U_3$.
        \item Compute $\dim \mathcal{H}_2(G) = \operatorname{nullity} \partial^C_2 - \operatorname{rank} \partial^U_3$.
    \end{enumerate}
    \item Run the cubical homology algorithm from \cite{kapuklin-kershaw:computations} on $G$, and compare the results with this new algorithm, which we call the ``basic shapes algorithm''. 
\end{enumerate}

Correctness of this algorithm is guaranteed by \cref{prop:basic_comm_diag}. We compare our results against those of the previously-developed cubical homology algorithm in order to benchmark the relative speedup obtained by only filling in these basic shapes. 

For our test graphs, we use Erd\H os-R\`enyi graphs. For an integer $n \geq 0$ and real number $p \in [0,1]$, an Erd\H os-R\`enyi graph $G(n,p)$ is a graph with $n$ vertices, with pairs of vertices connected by an edge independently with fixed probability $p$.

We use Erd\H os-R\`enyi graphs because by taking $p$ small (in this case, all of our test graphs satisfy $0.03 \leq p \leq 0.07$), we can ensure the presence of sparse graphs with a mix of trivial and nontrivial second discrete homology groups. After running both the basic shapes algorithm and the cubical algorithm on 200 test Erd\H os-R\`enyi graphs, all second homology dimensions matched, as expected. All test graphs had between 100 and 200 vertices, and since they are random graphs, their topology is random. Plots comparing the runtimes from this experiment are below.

\begin{center}
    \includegraphics[scale=0.5]{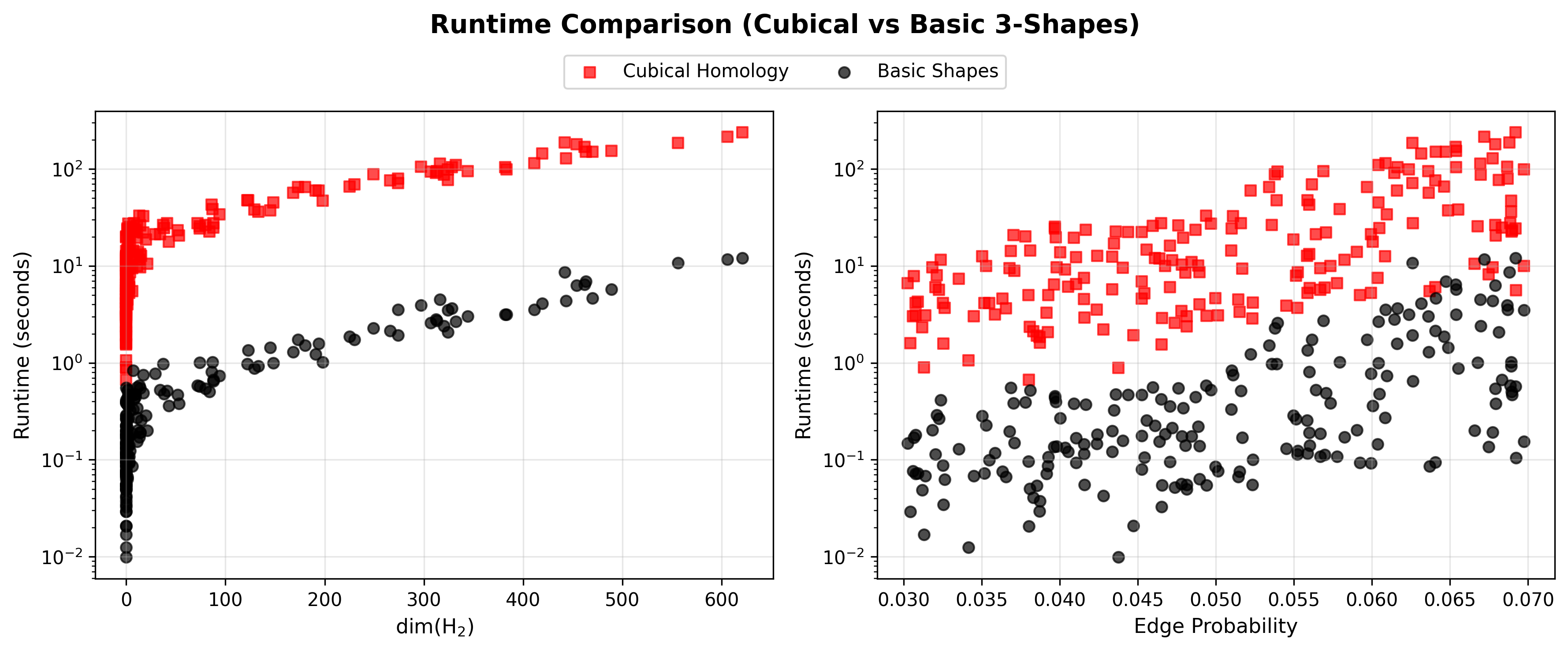}
\end{center}

The new algorithm outlined above was faster than the cubical algorithm from \cite{kapuklin-kershaw:computations} on each of the 200 test graphs, with a median ratio of 47.4532 between the runtime of the cubical algorithm and that of the basic shapes algorithm. All source code, plots, data files and a CSV file containing the full results of this experiment can be found at \url{https://github.com/JacobEnder/DiscreteHomology-Algorithms}. 

This experiment serves as a valuable proof of concept that our new procedure can indeed be used to considerably speed up computations of discrete homology, provided the appropriate chain groups are known in advance. As discussed, the primary limitation of this approach is computational --- it very quickly becomes impossible to precompute $T(n)$. However, this new procedure already gives a considerably faster algorithm for computing $\mathcal{H}_2(G)$.

% Uncomment the following if you have a bibliography file
 \bibliographystyle{amsalphaurlmod}
 \bibliography{all-refs.bib}

\end{document}